\documentclass[lettersize,journal]{IEEEtran}
\renewcommand{\baselinestretch}{1.0}
\usepackage{graphicx}
\usepackage{amsmath}
\usepackage{amsfonts}
\usepackage{amssymb}
\usepackage{amsthm}
\usepackage{color}
\usepackage{multicol}
\usepackage{multirow}
\usepackage{bm}
\usepackage{tabularx}
\usepackage{cite}
\usepackage{booktabs}
\usepackage{algorithm}
\usepackage{algpseudocode}
\usepackage{url}
\usepackage[T1]{fontenc}
\usepackage{tikz}  
\usetikzlibrary{shapes.geometric, arrows.meta, positioning, fit}

\newtheorem{theorem}{Theorem}
\newtheorem{lemma}{Lemma}
\newtheorem{corollary}{Corollary}

\theoremstyle{remark}

\newcommand{\Ybus}{\bm{Y}_{\mathrm{bus}}}
\newcommand{\Ybr}{\bm{Y}_{\mathrm{br}}}
\newcommand{\Sbus}{\bm{S}_{\mathrm{bus}}}
\newcommand{\Ibus}{\bm{I}_{\mathrm{bus}}}
\newcommand{\Vbus}{\bm{V}}
\newcommand{\Vnorm}{\hat{\bm{V}}}
\newcommand{\Jfull}{\bm{J}_{\mathrm{full}}}
\newcommand{\Jred}{\bm{J}_{\mathrm{red}}}
\newcommand{\Proj}{\mathcal{P}}
\newcommand{\diag}{\operatorname{diag}}
\newcommand{\struct}{\operatorname{struct}}

\begin{document}

\title{Symbolic-Permutation-Based Direct Sparse Jacobian Assembly for Accelerated Power-Flow Analysis}

\author{
  Tianshi Cheng,~\IEEEmembership{Member,~IEEE,}
  and~Venkata Dinavahi,~\IEEEmembership{Fellow,~IEEE}
  \thanks{This work was supported by the Natural Sciences and Engineering Research
  Council of Canada (NSERC) and Mitacs.}
  \thanks{\revision{The original manuscript was submitted in May~2026. This preprint includes updated results.}}
  \thanks{T.~Cheng is with the RTDS Technologies Inc., Winnipeg, MB R3T~2E1, Canada, and V.~Dinavahi is with the Department of Electrical and Computer
  Engineering, University of Alberta, Edmonton, AB T6G~2V4, Canada
  (e-mail: tcheng1@ualberta.ca; dinavahi@ualberta.ca).}
}

\newcommand{\revision}[1]{{\color{black}#1}}
\IEEEaftertitletext{%
  \begin{center}
    \begin{minipage}{0.96\textwidth}
      \centering\footnotesize
      \revision{This work has been submitted to the IEEE for possible publication.
      Copyright may be transferred without notice, after which this version may
      no longer be accessible.}
    \end{minipage}
  \end{center}
  \vspace{0.5\baselineskip}
}

\maketitle

\begin{abstract}
The repeated assembly of the sparse Jacobian matrix constitutes a primary computational bottleneck in the Newton--Raphson power-flow analysis. Through an algebraic equivariance analysis of the Jacobian matrix, this paper proposes a direct sparse Jacobian assembly method based on symbolic permutation. By pre-applying the permutation to the admittance matrix during a one-time symbolic phase, the runtime bus-type projection overhead is entirely eliminated. Consequently, complex Jacobian differentiation and global matrix assembly fuse into a single operation, yielding a single-pass, branchless algorithm that streams values directly into contiguous memory. The framework is realized in the open-source \texttt{RustPower} solver. Benchmarks on IEEE test systems and PEGASE 9241 demonstrate that the proposed assembly is \revision{$24.4$--$30.1\times$} faster than the MATPOWER-style baseline and \revision{$6.0$--$7.8\times$} faster than the Sch\"{a}fer--Braun optimized assembly. End-to-end, the \revision{grid-level} solve time is \revision{$11.7\times$--$154.3\times$ faster than pandapower~3.5.4 with Numba}, while leaving numerical results and convergence behavior unchanged. \revision{Compared with the native C++ LightSim2Grid~1.0.0, grid-level performance is comparable on IEEE~39 and $1.3\times$--$1.8\times$ faster on IEEE~118 and the PEGASE systems.} The method extends directly to general problems where a sparse Jacobian inherits its structure from a static network operator, including transient stability analysis, EMT simulation, and OPF Hessian assembly.
\end{abstract}


\section*{Nomenclature}
\begin{table}[htbp]
  \centering
  \begin{tabularx}{\columnwidth}{@{}p{1.55cm}X@{}}
    $n,\;b$            & Number of buses; number of branches. \\
    $\bm{A}$           & Branch--bus incidence matrix. \\
    $\Ybr,\;\Ybus$     & Primitive branch admittance; nodal admittance matrix. \\
    $\Vbus,\;\Ibus$    & Complex bus voltage and injected-current vectors. \\
    $|\Vbus|,\;\bm{\theta}$ & Bus voltage magnitude and angle vectors. \\
    $\Vnorm$           & Unit-modulus voltage, $\Vnorm=e^{j\bm{\theta}}$. \\
    $\Sbus$            & Complex bus power injection, $\bm{P}+j\bm{Q}$. \\
    $\Sbus^{\mathrm{sp}}$ & Specified (scheduled) bus power injection. \\
    $\bm{x},\;\bm{x}_{\mathrm{full}}$ & State vector of unknowns; full $2n$ state. \\
    $\bm{F}$           & Power-mismatch (residual) vector. \\
    $\nu$              & Newton--Raphson iteration index. \\
  \end{tabularx}
  \medskip
  \begin{tabularx}{\columnwidth}{@{}p{1.55cm}X@{}}
    $\Jfull,\;\Jred$   & Full and bus-type-reduced power-flow Jacobian. \\
    $\Proj$            & Reduction slicer $[\bm{I}_{N_{\text{red}}}\mid\bm{0}]$. \\
    $\bm{T}$           & Orthogonal bus-type permutation matrix. \\
    $\struct(\cdot)$   & Boolean (0/1) sparsity matrix of a matrix. \\
    $B_Y,\;B_J$        & Boolean sparsity matrices of $\Ybus$ and $\Jred$. \\
    $N_{pv},\;N_{pq}$  & Number of PV and PQ buses. \\
    $N_{\text{red}}$   & $N_{pv}+2N_{pq}$, the reduced-system dimension. \\
    NR                 & Newton--Raphson. \\
    CSC, CSR           & Compressed sparse column / row storage. \\
    ECS                & Entity--component--system. \\
  \end{tabularx}
\end{table}

\section{Introduction}
\label{sec:intro}
\IEEEPARstart{S}{teady} power-flow analysis is the computational kernel beneath planning, operations, contingency screening, and optimization of electric power systems. 
The Newton--Raphson (NR) method has dominated this task since its introduction by Tinney and Hart in 1967~\cite{tinney1967power}, together with optimally ordered sparse factorization~\cite{tinney1967direct} establishing the direct-solver paradigm that remains the standard today. 
At each Newton step, the Jacobian matrix is assembled and the resulting sparse linear system is solved by direct methods such as KLU~\cite{davis2010klu}, which separate a one-time symbolic analysis from repeated numeric re-factorization, making the latter an allocation-free, cache-efficient pipeline. 
For extremely large or ill-conditioned systems where the Jacobian condition number degrades and direct factorization becomes numerically unstable, the inner linear system may instead be solved by preconditioned GMRES or BiCGStab~\cite{idema2012scalable,yetkin2023recycling}; 
however, the pre-conditioner still requires the explicit Jacobian structure, and for standard benchmarks these Newton--Krylov variants remain slower than a well-tuned direct solver, leaving LU-based NR as the workhorse in practice. Robustness extensions---Iwamoto--Tamura damping~\cite{iwamoto1981load}, Levenberg--Marquardt regularization~\cite{lagace2008improving,pires2019robust}, and continuous Newton methods~\cite{milano2009open,yu2025cfn}---addressed convergence fragility at the cost of more outer iterations and hence more frequent Jacobian evaluations.

Consequently, the assembly of the Jacobian matrix---an indispensable step in every Newton-based iteration---has superseded sparse factorization as the dominant computational cost.
Classical textbooks often formulated the Jacobian by element-wise differentiation of the scalar polar power-flow equations, incurring $\Theta(nnz)$ transcendental function evaluations and tedious index bookkeeping for each entry~\cite{bergen2000power}. 
Zimmerman et al.~\cite{zimmerman2011matpower} revolutionized this practice by maintaining the algebraic integrity of the network entirely in the complex domain through Wirtinger Calculus. 
The resulting matrix-differential formulation collapses the transcendental burden from a traditional $\Theta(nnz)$ element-wise sweep to a single $\Theta(n)$ vector map per iteration, and it became the foundational paradigm for modern power-flow software. 
This complex-domain approach was ported to Python as PYPOWER~\cite{pypower}, extended into an object-oriented modeling environment in pandapower~\cite{thurner2018pandapower}, and reimplemented with a C++/Eigen backend in LightSim2Grid~\cite{lightsim2grid} for native speed.

In 2018, Sch\"{a}fer and Braun~\cite{schaefer2018jacobian} were among the first to systematically identify the Jacobian-assembly bottleneck, replacing MATPOWER's block-stacking pipeline with an optimized CSR formation fused by Numba JIT compilation, reporting 3--14$\times$ assembly speedups compared to MATPOWER. The C++ LightSim2Grid~\cite{lightsim2grid} also adopted a similar element-level assembly approach.
Subsequent work has pushed the parallel frontier: heterogeneous CPU--GPU solvers exploit invariant sparsity across ensemble members to amortize symbolic factorization over thousands of N--1 or Monte-Carlo cases~\cite{wang2021fast,zhou2017gpu,zhou_gpuPPF} as well as large-scale optimization problems~\cite{ryu_gpuopf}. 

Yet all of these works and advances---whether in MATLAB, Python, C++, or CUDA---share a common architectural pitfall:
they remain trapped in the three-step ``differentiate, partition, assemble'' pipeline of Jacobian assembly for Newton iteration methods.
Differentiation evaluates full $2n \times 2n$ complex sensitivities with thousands of unused computations of PV and slack buses; 
partitioning scatters entries through pointer-chasing and dynamic indexing; 
assembly allocates temporary blocks and copies them into the global Jacobian. 
These structural overheads dominate the \revision{iteration loop}---\revision{71--85\%} of per-iteration time (Fig.~\ref{fig:breakdown})---and they seem to be \emph{fundamental}: no engineering tuning can eliminate the three steps in the current open-source implementations.

As shown in Fig.~\ref{fig:overview}, this paper proposes a symbolic-permutation-based direct sparse Jacobian assembly method that significantly reduces the iteration structural overhead in Newton--Raphson power-flow analysis. The work consists of three main contributions:
\begin{enumerate}
\item \textbf{Permutation-equivariance theorem.} 
Theorem~\ref{thm:iso} proves that the Jacobian assembly commutes with orthogonal bus permutation, identifying the bus-type projection as the algebraic root cause of block-assembly fragmentation and enabling its hoisting to a one-time symbolic phase~(Section~\ref{sec:symbolic}).

\item \textbf{Single-pass direct-fill algorithm.} 
With the Jacobian structure fixed by Corollaries~\ref{cor:reduced}--\ref{cor:boolean}, per-iteration assembly reduces to a branch-free, allocation-free $\Theta(nnz)$ sweep: each column of the admittance matrix is differentiated directly into its final position in the reduced Jacobian, eliminating intermediate sensitivity blocks, per-iteration sparse structure reconstruction, and wasted derivative evaluations for dropped states~(Section~\ref{sec:symbolic}-Section~\ref{sec:impl}).

\item \textbf{Open-source \texttt{RustPower} implementation.}
The method (noted as V2) is implemented in \texttt{RustPower}, an open-source data-oriented ECS framework.
Isolated assembly benchmarks demonstrate \revision{$24.4\times$--$30.1\times$} speedup over baseline MATPOWER-style~\cite{zimmerman2011matpower} formulations (V0) and \revision{$6.0\times$--$7.8\times$} over the optimized pandapower-style algorithm~\cite{schaefer2018jacobian} (V1), eliminating \revision{$95.9\%$--$96.7\%$} of the baseline assembly overhead (and up to \revision{$87.1\%$} of the residual overhead in V1). 
\revision{End-to-end ablation confirms $3.2\times$--$4.5\times$ overall solve speedup over V0 and $1.6\times$--$2.0\times$ over V1 with the same KLU backend~(Section~\ref{sec:results}).}
\end{enumerate}
 
Section~\ref{sec:conclusion} concludes the paper and discusses future extensions to transient stability analysis, EMT simulation, and OPF Hessian assembly.
 
\section{Mathematical Formulation}
\label{sec:formulation}

\subsection{Network and Admittance Matrix}
In steady-state analysis, a power network of $n$ buses and $b$ branches is described by its branch--bus incidence matrix
$\bm{A}$; where $\Ybr=\operatorname{blkdiag}(y_1,\dots,y_b)$ is the branch
admittance matrix, the nodal admittance matrix can be expressed as 
\begin{equation}
  \Ybus=\bm{A}^{\!\top}\Ybr\,\bm{A}.
  \label{eq:ybus}
\end{equation}
The network matrix $\Ybus$ fundamentally inherits the properties of a weighted Laplacian matrix.

Buses are ordered as $[\,\mathrm{PQ}\mid\mathrm{PV}\mid\mathrm{slack}\,]$, by the number of state variables each bus contributes: PQ buses carry both $\theta$ and $|V|$ as states and lead, PV buses carry $\theta$ only and follow, slack buses carry no state and close the list. The ``missing state'' types therefore concentrate on the right. The ordering is realised by a one-time permutation $\bm{T}$ applied before iteration; the pre-ordered admittance matrix $\bm{T}\Ybus\bm{T}^{\top}$ locks the sparse topology, and every subsequent block split coincides with a contiguous row range.

\subsection{The Power-Flow Equation}
The complex bus power injection is the nonlinear function of the bus voltage given by
\begin{equation}
  \Sbus(\Vbus)=\diag(\Vbus)\,\overline{\Ybus\Vbus},
  \label{eq:sbus}
\end{equation}
with $\overline{(\cdot)}$ complex conjugation. Expanded element-wise, \eqref{eq:sbus} becomes the classical polar
form
\begin{align}
P_i &= v_i \sum_{k \sim i} v_k \bigl( G_{ik}\cos\theta_{ik} + B_{ik}\sin\theta_{ik} \bigr),   \label{eq:p_scalar} \\
Q_i &= v_i \sum_{k \sim i} v_k \bigl( G_{ik}\sin\theta_{ik} - B_{ik}\cos\theta_{ik} \bigr), \label{eq:q_scalar}
\end{align}
where $\theta_{ik} = \theta_i - \theta_k$, $G_{ik} + jB_{ik}$ denotes the $(i,k)$-th element of the admittance matrix $\Ybus$, and the summation runs over all buses $k$ connected to bus $i$.

Although these scalar formulations are common in textbooks to directly differentiate and build the Jacobian, this representation is highly inefficient. Every individual entry in the resulting Jacobian inherits  distinct transcendental pairs $(\sin\theta_{ik}, \cos\theta_{ik})$, forcing conventional element-wise differentiation to be a slow $\Theta(\mathrm{nnz})$ operation per iteration.

This heavy algebraic penalty historically obscured the potential of efficient direct Jacobian assembly. Consequently, modern tools adopted complex matrix differentiation to bypass the high transcendental function cost.
 \begin{figure*}[t]
\centering
  \includegraphics[width=2\columnwidth]{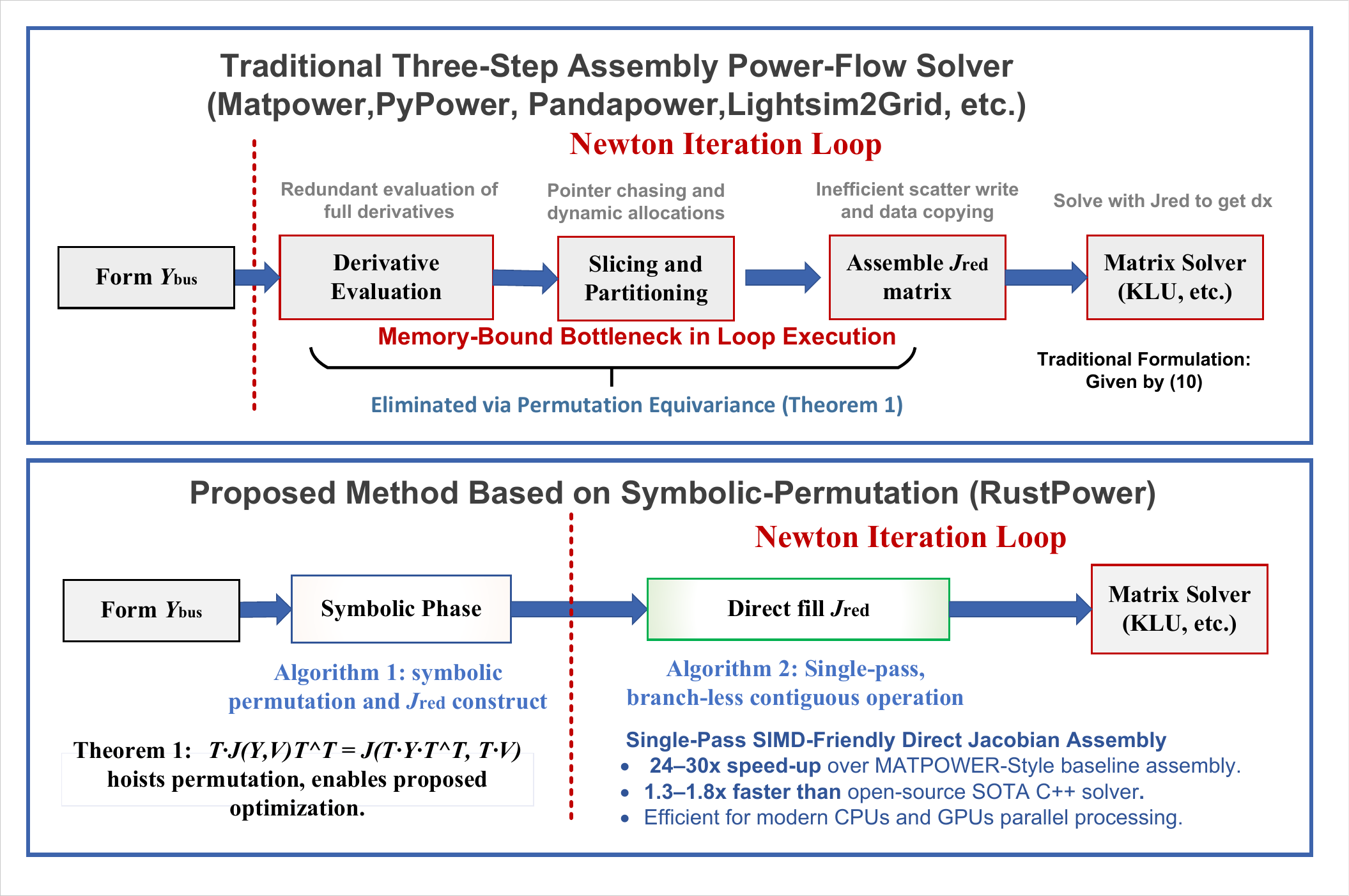}
\caption{Overview of the proposed method's core innovation. By hoisting topology-dependent permutations offline via Theorem 1, the proposed framework eliminates runtime slicing and scattered writes, condensing the traditional loop into a single-pass direct-fill stream to achieve up to 21$\times$ speedup and outperform the state-of-the-art open-source C++ solver.}
\label{fig:overview}
\end{figure*}
\subsection{Analytical Matrix Differentiation and Jacobian Assembly}
 
The Newton--Raphson method eliminates the non-linear power mismatch by iteratively solving the linearized system
\begin{equation}
\bm{J}^{(\nu)}\Delta\bm{x}^{(\nu)}=-\bm{F}\bigl(\bm{x}^{(\nu)}\bigr), \quad \bm{x}^{(\nu+1)}=\bm{x}^{(\nu)}+\Delta\bm{x}^{(\nu)},
\label{eq:nr}
\end{equation}
where $\nu$ is the iteration counter. In this holistic formulation, the global residual vector $\bm{F}(\bm{x}) = \bm{0}$ is directly mapped from the complex algebraic network constraint
\begin{equation}
 \bm{F}(\Vbus) = \Sbus^{\mathrm{sp}} - \diag(\Vbus)\overline{\Ybus\Vbus} = \bm{0},
\label{eq:deltaS}
\end{equation}

and the Jacobian is expressed in matrix form as
\begin{equation}
  \Jfull=\begin{bmatrix}\partial\Sbus/\partial\bm{\theta}  &
  \partial\Sbus/\partial|\Vbus|\end{bmatrix}  ,
  \label{eq:jfull}
\end{equation}
where $\Sbus^{\mathrm{sp}}$ contains the constant constraints (full $\bm{S}$ for PQ buses, $P$ only for PV buses), and $\Vbus$ has constant constraints ($|\bm{V}|$ for PV, $\bm{V}\angle\theta$ for slack buses).
In their pioneering framework, Zimmerman et al.~\cite{zimmerman2011matpower} introduced an inspiring, mathematically elegant matrix-differential method for MATPOWER. Instead of expanding the governing power equations into scalar polar loops prematurely, they maintained the algebraic integrity of the network entirely within the complex domain through the continuous exponential voltage alignment $\bm{V} = \diag(|\bm{V}|)\bm{\hat{V}}$ with $\bm{\hat{V}} = e^{j\bm{\theta}}$. By exploiting the compact network current dependency $\bm{I}_{\mathrm{bus}} = \Ybus \bm{V}$, this matrix-differential formulation directly yields the exact complex sensitivity operators:
\begin{align}
\frac{\partial\bm{S}}{\partial\bm{\theta}} &= j\diag(\bm{V})\bigl(\diag(\overline{\bm{I}_{\mathrm{bus}}}) - \overline{\Ybus}\diag(\overline{\bm{V}})\bigr), \label{eq:dStheta}\\
\frac{\partial\bm{S}}{\partial|\bm{V}|} &= \diag(\bm{V})\overline{\Ybus}\diag(\overline{\bm{\hat{V}}}) + \diag(\overline{\bm{I}_{\mathrm{bus}}})\diag(\bm{\hat{V}}). \label{eq:dSvm}
\end{align}
This holistic complex-domain formulation entirely eliminates the heavy trigonometric evaluations from the differentiation phase, successfully collapsing the transcendental computational burden from a traditional $\Theta(\mathrm{nnz})$ down to a single node-level $\Theta(n)$ vector map per iteration.

To mitigate the heavy computational burden of sparse matrix-vector products in conventional complex-domain formulations, Sch\"{a}fer and Braun~\cite{schaefer2018jacobian} brilliantly deployed Numba JIT compilation in pandapower to fuse these underlying operations. By mapping elements directly based on the structural invariance of the network topology, their approach successfully streamlined derivative evaluations and bypassed some of the intermediate matrix products, delivering an impressive 3$\times$--14$\times$ speedup.

However, while this element-wise fusion successfully accelerated the differentiation phase, the underlying solver remains trapped in a classic architectural pitfall. Because the complex power-flow equations are non-holomorphic, conventional implementations must maintain extensive layout-mapping tables to partition derivatives into discrete $J_{11}$ through $J_{22}$ sub-blocks. This structural dependency burdens the \revision{iteration loop} with continuous pointer-dereferencing and non-contiguous scatter writes that completely disrupt hardware prefetching.

To resolve this architectural pitfall, the power-flow equations demand further mathematical analysis.

\section{Symbolic Matrix Analysis of Domain Projection and Reordering}
\label{sec:symbolic}
 
\subsection{Non-Holomorphic Deadlock and the Dual Real-Mapping Paradigm}

The fundamental mathematical bottleneck of conventional power-flow algorithms originates from a non-holomorphic deadlock in the complex domain. The governing complex power injection equation $\Sbus(\Vbus) = \diag(\Vbus)\overline{\Ybus\Vbus}$ violates the Cauchy-Riemann conditions due to the non-differentiable nature of the conjugation operator, forcing a paradigm shift that maps the complex sensitivities onto a dual real-coordinate system
\begin{align}
\bm{J}_{\mathrm{red}} & = \mathcal{P} \left( \bm{T}_{\mathrm{block}} \bm{J}_{\mathrm{full}}(\bm{Y}_{\mathrm{bus}}, \bm{V}) \bm{T}_{\mathrm{block}}^{\top} \right) \mathcal{P}^{\top} \\ \notag
& = \mathcal{P} \begin{bmatrix} \bm{T} & \bm{0} \\ \bm{0} & \bm{T} \end{bmatrix} \begin{bmatrix} \Re\left(\frac{\partial\bm{S}}{\partial\bm{\theta}}\right) & \Re\left(\frac{\partial\bm{S}}{\partial|\bm{V}|}\right) \\   \Im\left(\frac{\partial\bm{S}}{\partial\bm{\theta}}\right) & \Im\left(\frac{\partial\bm{S}}{\partial|\bm{V}|}\right) \end{bmatrix} \begin{bmatrix} \bm{T}^{\top} & \bm{0} \\ \bm{0} & \bm{T}^{\top} \end{bmatrix} \mathcal{P}^{\top}.
\end{align}

Under this traditional real-mapping paradigm, contemporary implementations—including state-of-the-art automatic differentiation (AutoDiff) frameworks—are bound by a conventional core assumption: that the unconstrained global matrix $\bm{J}_{\mathrm{full}}$ must be explicitly materialized prior to handling any spatial configurations. Consequently, traditional power-flow solvers and AutoDiff engines are restricted to evaluating the full, unconstrained matrix, remaining fundamentally agnostic to external boundary constraints and spatial node partitions, which inherently forces conventional architectures to execute three computationally catastrophic operations inside every \revision{iteration of the loop}:
\begin{enumerate}
    \item \textbf{Redundant Differentiation:} Wasting significant CPU cycles and arithmetic throughput evaluating derivatives for dropped states (e.g., slack and PV bus voltage magnitudes) that are subsequently discarded.
  \item \textbf{Indirect Spatial Partitioning:} Relying on extensive layout-mapping tables to index and route derivatives into traditional $\bm{J}_{11}$ through $\bm{J}_{22}$ sub-blocks, introducing severe pointer-dereferencing and indirection overheads inside the \revision{iteration loop}.
\item \textbf{Non-Contiguous Scatter Assembly:} Enforcing dimensional reduction via scattered writes that distribute values across disjoint matrix segments, which completely disrupts hardware prefetching, invalidates CPU cache lines, and severely limits parallel scalability.
\end{enumerate}

These limitations transform the iteration hot path into a performance graveyard of branch mispredictions and cache invalidations, which also significantly lower the parallel efficiency on modern CPUs and GPUs.

\subsection{Permutation Equivariance of the Sensitivity}
The complex sensitivities~\eqref{eq:dStheta}--\eqref{eq:dSvm} are built from three
primitive operations on $\Ybus$ and $\Vbus$: diagonal scaling, matrix product, and
complex conjugation. Each of them commutes with conjugation by an orthogonal
permutation. The elementary fact at the bottom of the proof is the following

\begin{lemma}[Diagonal--permutation commutation]
\label{lem:diag}
For any vector $\bm{x}$ and orthogonal permutation matrix $\bm{T}$ with
$\bm{T}\bm{T}^{\!\top} = \bm{I}$,
\begin{equation}
  \bm{T}\,\diag(\bm{x})\,\bm{T}^{\!\top} \;=\; \diag(\bm{T}\bm{x}).
  \label{eq:diag-perm}
\end{equation}
\end{lemma}
\begin{proof}
Let $\pi$ be the permutation with $T_{ij} = \delta_{i,\pi(j)}$. Then
\begin{align}
  \bigl(\bm{T}\diag(\bm{x})\bm{T}^{\!\top}\bigr)_{ij}
  & =  \sum_{k} \delta_{i,\pi(k)}\,\delta_{j,\pi(k)}\,x_k   
  & = \delta_{ij}\,x_{\pi^{-1}(i)} \notag \\
  & = \bigl(\diag(\bm{T}\bm{x})\bigr)_{ij}. 
\end{align}
\end{proof}

\begin{theorem}[Permutation equivariance of the Jacobian]
\label{thm:iso}
Let $J(\bm{Y},\bm{V})$ denote either complex sensitivity matrix of
\eqref{eq:dStheta}--\eqref{eq:dSvm}. For any orthogonal permutation
$\bm{T}\in\{0,1\}^{n\times n}$,
\begin{equation}
  \bm{T}\,J(\bm{Y},\bm{V})\,\bm{T}^{\!\top}
  \;=\;
  J\!\bigl(\bm{T}\bm{Y}\bm{T}^{\!\top},\,\bm{T}\bm{V}\bigr).
  \label{eq:equivariance}
\end{equation}
\end{theorem}
\begin{proof}
Each primitive operation in $J$ commutes with $\bm{T}$-conjugation:
\begin{itemize}
  \item matrix product, $\bm{T}\bm{A}\bm{B}\bm{T}^{\!\top}
        = (\bm{T}\bm{A}\bm{T}^{\!\top})(\bm{T}\bm{B}\bm{T}^{\!\top})$ since
        $\bm{T}^{\!\top}\bm{T} = \bm{I}$;
  \item diagonalisation, $\bm{T}\diag(\bm{x})\bm{T}^{\!\top} = \diag(\bm{T}\bm{x})$ by
        Lemma~\ref{lem:diag};
  \item complex conjugation, $\bm{T}\,\overline{\bm{A}}\,\bm{T}^{\!\top}
        = \overline{\bm{T}\bm{A}\bm{T}^{\!\top}}$ since $\bm{T}$ is real.
\end{itemize}
Apply these to \eqref{eq:dStheta}--\eqref{eq:dSvm} entrywise, using
$\bm{T}\bm{I}_{\mathrm{bus}}(\bm{Y},\bm{V})
 = (\bm{T}\bm{Y}\bm{T}^{\!\top})(\bm{T}\bm{V})
 = \bm{I}_{\mathrm{bus}}(\bm{T}\bm{Y}\bm{T}^{\!\top},\bm{T}\bm{V})$.
\end{proof}

 $\bm{J}_{\mathrm{full}}$ inherits this equivariance because the transformation $\bm{T}$ is real. We select a state-level permutation that concentrates all active variables $[\bm{\theta}_{\mathrm{PQ}}; \bm{\theta}_{\mathrm{PV}}; |\bm{V}|_{\mathrm{PQ}}]$ at the front, while trailing the omitted PV and slack states at the end. The reduction operator simplifies to a trivial memory slicer, completely eliminating runtime coordinate projections.

\begin{corollary}[Pre-permutation absorbs the projection]
\label{cor:reduced}
Let $\mathcal{P}=[\,\bm{I}_{a}\;\,\bm{0}\,]$ denote the static truncation matrix. 
The row-and-column permutation $\bm{T}$ isolates the active variables, partitioning the transformed Jacobian into a block layout:
\begin{equation*}
  \bm{T}\bm J_{\mathrm{full}}\bm{T}^{\!\top}
  \;=\;
  \begin{bmatrix} \bm{J}_{aa} & \bm{J}_{ad} \\ \bm{J}_{da} & \bm{J}_{dd} \end{bmatrix},
\end{equation*}
where $a$ and $d$ denote the active and dropped domains. Because the block structuring forces the projection to collapse into identity matrices, the truncation operator $\mathcal{P}$ directly absorbs the transformation. Combining this layout with Theorem~\ref{thm:iso} yields:
\begin{equation}
  \Jred
  \;=\;
  \mathcal{P}\,\Jfull\!\bigl(\bm{T}\Ybus\bm{T}^{\!\top},\,\bm{T}\Vbus\bigr)\,\mathcal{P}^{\!\top} \;\equiv\; \bm{J}_{aa}.
  \label{eq:absorb}
\end{equation}
The bus permutation $\bm{T}$ is applied once before execution. Due to this block alignment, the runtime projection is entirely absorbed into the source matrix, leaving the solver to extract the top-left block via a single memory slice.
\end{corollary}

\begin{figure}[t]
  \centering
  \includegraphics[width=\columnwidth]{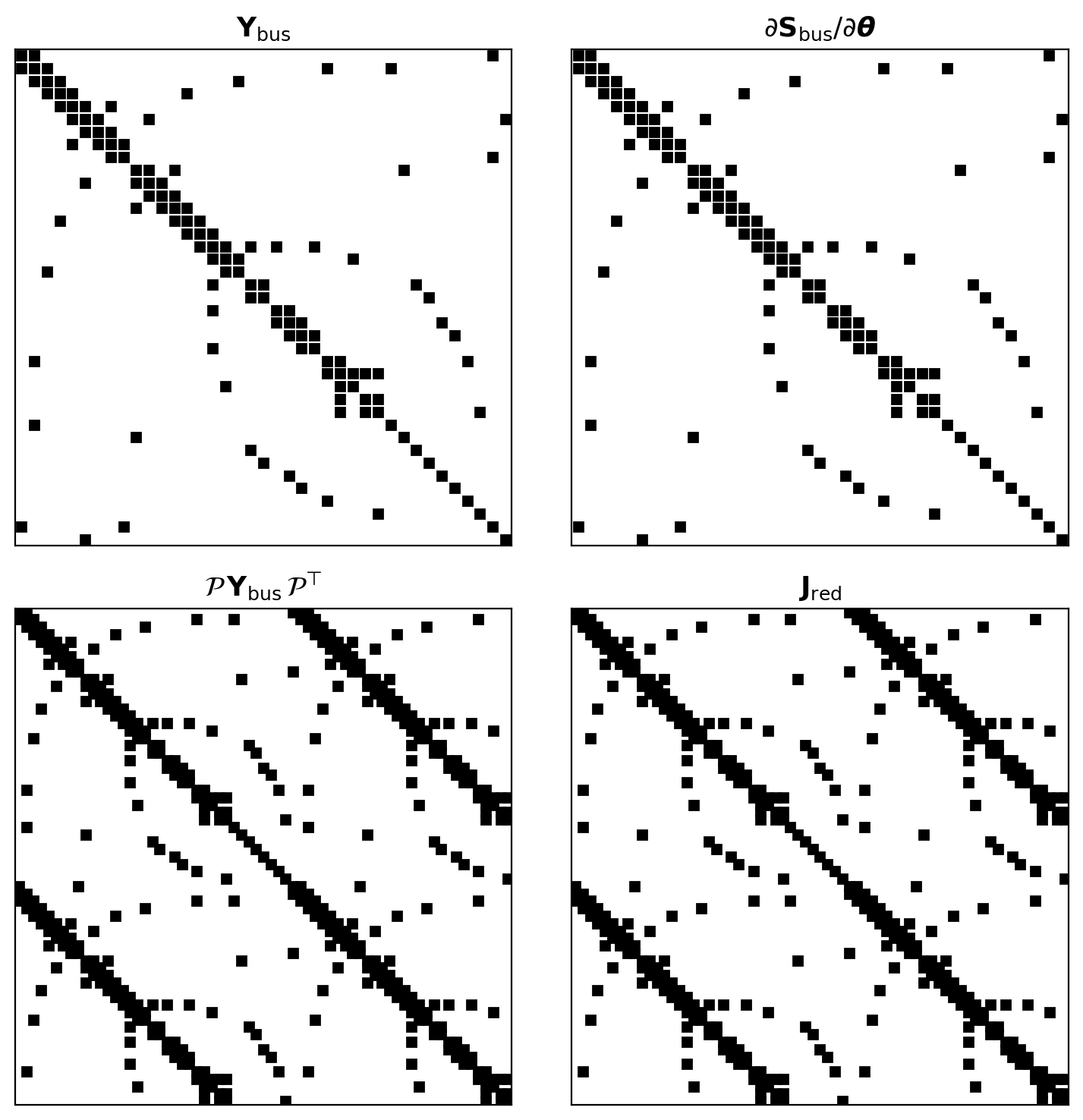}
  \caption{Permutation equivariance and reduction on IEEE 39, buses in
  $[\mathrm{PQ}\mid\mathrm{PV}\mid\mathrm{slack}]$ order. \emph{Top:} the complex
  sensitivity $\partial\Sbus/\partial\bm{\theta}$ inherits $\Ybus$'s nonzero positions
  (Theorem~\ref{thm:iso}; $\partial\Sbus/\partial|\Vbus|$ is identical, omitted for
  space). \emph{Bottom:} $\Jred$ is the top-left $(N_{pv}+2N_{pq})$-square block of the
  full Jacobian assembled on the pre-permuted admittance matrix
  $\bm{T}\Ybus\bm{T}^{\!\top}$ (Corollary~\ref{cor:boolean}).}
  \label{fig:isomorphism}
\end{figure}
\begin{corollary}[Boolean structural mapping]
\label{cor:boolean}
Apply the structural operator $\struct(\cdot)$ to \eqref{eq:absorb}. Because the projection $\mathcal{P}$ acts as a binary slicer, the operator passes through it:
\begin{equation}
  \struct(\Jred) \;=\; \mathcal{P}\,\struct\bigl(\Jfull(\bm{T}\Ybus\bm{T}^{\!\top},\,\bm{T}\Vbus)\bigr)\,\mathcal{P}^{\!\top}.
  \label{eq:boolean-slicer}
\end{equation}
With $B_Y = \struct(\Ybus)$, the real-linear split duplicates the pattern into four quadrant sub-blocks:
\begin{equation}
  B_{\Jfull} \;=\; \struct(\Jfull) \;=\;
  \begin{bmatrix} B_Y & B_Y \\ B_Y & B_Y \end{bmatrix} \in \{0,1\}^{2n\times 2n},
  \label{eq:bjfull}
\end{equation}
and the equivariance carries this pattern through to the permuted matrix, $\struct(\Jfull(\bm{T}\Ybus\bm{T}^{\!\top},\,\bm{T}\Vbus)) = \bm{T}\,B_{\Jfull}\,\bm{T}^{\!\top}$. Hence,
\begin{equation}
  \struct(\Jred) \;=\; \mathcal{P}\,\bm{T}\,B_{\Jfull}\,\bm{T}^{\!\top}\,\mathcal{P}^{\!\top}.
  \label{eq:boolean-similarity}
\end{equation}
\end{corollary}

 Fig.~\ref{fig:isomorphism} shows the IEEE 39-bus system; the matrix patterns confirm Theorem~\ref{thm:iso} and Corollaries~\ref{cor:reduced}--\ref{cor:boolean}.
\subsection{Direct Closed-Form Address Mapping}
By expanding the analytical complex sensitivities \eqref{eq:dStheta}--\eqref{eq:dSvm} along individual columns, the exact structural column locality evaluates by inspection:
\begin{align}
\Bigl[\tfrac{\partial\Sbus}{\partial\bm{\theta}}\Bigr]_{:,k}
&=j\,\diag(\Vbus)\bigl(\overline{I_k}\,\bm{e}_k
-\overline{V_k}\,\overline{[\Ybus]_{:,k}}\bigr), \label{eq:coltheta} \\ \notag
\Bigl[\tfrac{\partial\Sbus}{\partial|\Vbus|}\Bigr]_{:,k}
&=\overline{\hat V_k}\diag(\Vbus)\overline{[\Ybus]_{:,k}}
+\hat V_k\overline{I_k}\,\bm{e}_k.
\end{align}
where $\bm{e}_k$ is the $k$-th unit vector. 
 The equations establish that the numerical derivation of the $k$-th column of the operational Jacobian requires the same $k$-th column of $\Ybus$, leading to direct address mapping instead of generic sparse matrix indices.

Since \eqref{eq:coltheta} executes strictly within decoupled column scopes, the nonzero destination addresses inside the real operational Jacobian $\bm{J}_{\mathrm{red}}$ are computed directly from the static topology $\text{struct}(\bm{Y}_{\mathrm{bus}})$ without matrix traversal. For any nonzero branch $(i,k) \in \text{struct}(\bm{Y}_{\mathrm{bus}})$, its pre-allocated NNZ memory locations can be evaluated via a unified linear step stride

Phase 1 ($\theta$-state column variations, $k \in [0, N_{\mathrm{pv}}+N_{\mathrm{pq}}-1]$):\begin{equation}\sigma_{\theta}(i, k) = c_{\theta_k} + \mathbb{I}(i \in \mathcal{I}_{\mathrm{PQ}}) \cdot a_k + \rho(i, k)\end{equation}

Phase 2 ($|\bm{V}|$-state column variations, $k \in [0, N_{\mathrm{pq}}-1]$):\begin{equation}\sigma_{|V|}(i, k) = c_{|V|k} + \mathbb{I}(i \in \mathcal{I}_{\mathrm{PQ}}) \cdot a_k + \rho(i, k)\end{equation}
The exact memory coordinates for all four implicit Jacobian sub-blocks are dynamically unified by three static topological attributes: the base column pointer $c_{(\cdot)}$, the zero-based row offset $\rho(i,k)$ inside the current column slice, and the isolated active nonzero cardinality $a_k$. The architectural and algebraic superiorities of this mapping are summarized as follows (Table~\ref{tab:placement}):

\textbf{Single-Pass Contiguous Writes}: The completeness ordering $\bm{T}$ isolates all boundary-constrained dropped nodes into the matrix lower-right corner. Consequently, the conditional indicator $\mathbb{I}(\cdot)$ translates into a monotonically increasing linear offset, forcing the hardware execution thread to automatically skip dropped $\mathrm{PV}$ and slack constraints while maintaining uninterrupted single-pass 1D streaming writes.

\textbf{Layout Generality}: Although this single-pass streaming kernel is presented column-wise under the Compressed Sparse Column (CSC) format, the underlying memory dispatch logic generalizes symmetrically to Compressed Sparse Row (CSR) layouts. The global target memory addresses remain entirely predictable a priori across rows, guaranteeing maximum execution throughput on modern SIMD-driven processors.
\begin{table}[htbp]
  \caption{Consequences of projection placement.}
  \label{tab:placement}
  \centering
  \begin{tabularx}{\columnwidth}{@{}lXX@{}}
    \toprule
     & Conventional / \cite{schaefer2018jacobian} & This work \\
    \midrule
    Placement of $T$       & $\bm J_{full}$ & $\Ybus$ \\
    $T(\cdot)T^{\!\top}$ in loop & every iteration & never \\
    Matrix construction        & block assembly  & direct contiguous write \\
    Write pattern              & scattered       & sequential, contiguous \\
    Storage order              & CSR             & CSC \\
    \bottomrule
  \end{tabularx}
\end{table}

\subsection{The Numeric Fill Operator}
With the pattern fixed by \eqref{eq:boolean-similarity}, the numeric phase only writes values from \eqref{eq:dStheta}--\eqref{eq:dSvm} to the matrix. By applying the general evaluation operators of those formulas uniformly to every entry, including the diagonal, we obtain the intermediate matrices $\mathcal{U}_\theta$ and $\mathcal{U}_v$. The $\diag(\overline{\Ibus})$ terms are strictly diagonal; thus, $\mathcal{U}_\theta$ and $\mathcal{U}_v$ coincide with the exact derivatives off the diagonal, while the diagonal entries can be independently corrected via localized linear patches based on
\begin{equation}
\tfrac{\partial\Sbus}{\partial\bm{\theta}}
=\mathcal{U}_\theta+j\,\diag(\Sbus^{(\nu)}),
\tfrac{\partial\Sbus}{\partial|\Vbus|}
=\mathcal{U}_v+\diag!\bigl(\Sbus^{(\nu)}!\oslash|\Vbus^{(\nu)}|\bigr).
\label{eq:kernel}
\end{equation}
Per-iteration assembly is therefore reduced to a branch-free $\Theta(\mathrm{nnz})$ sweep—evaluating a uniform relation for every nonzero entry without any $i\stackrel{?}{=}k$ tests—followed by a $\Theta(n)$ constant-time diagonal modification.

\section{Implementation in RustPower}
\label{sec:impl}

\subsection{Data-Oriented Architecture}
\texttt{RustPower} realizes the proposed method in Rust~\cite{matsakis2019rust} on a data-oriented
ECS architecture~\cite{bevyengine,nystrom2014data}, a design carried over from the prior real-time ECSGrid work~\cite{cheng2023ecsgrid}. All power system entities are stored in statically typed, columnar arrays in ECS archetype tables, where archetypes are similar to \textit{Bus} or \textit{Line} tables in pandapower.
 Networks are
read in the pandapower JSON/CSV format for reproducibility, and the linear solve
is delegated, behind a single \texttt{Solve} trait, to an interchangeable backend such as KLU~\cite{davis2010klu}. The proposed algorithm is not tied to \texttt{RustPower} or Rust and is applicable to other implementations without significant modifications.

\subsection{Sparse Storage and Contiguity of Assembly}
The column-decoupled nature of \eqref{eq:coltheta} aligns natively with the Compressed Sparse Column (CSC) layout. 
By eliminating intermediate representations, column $k$ of $\bm{Y}_{\mathrm{bus}}$ is streamed as a contiguous memory block to directly populate the corresponding columns of $\bm{J}_{\mathrm{red}}$ in a single pass. 
This strictly sequential data flow guarantees optimal spatial locality, enabling the hardware pre-fetcher to function ideally and effectively eliminating cache steering bottlenecks~\cite{drepper2007what}.

While the underlying mapping logic can be applied symmetrically to Compressed Sparse Row (CSR) layouts, state-of-the-art sparse direct solvers like KLU~\cite{davis2010klu} inherently execute factorization column-wise and strictly mandate CSC inputs. 
To exploit this architectural alignment, \texttt{RustPower} maintains a unified CSC pipeline across the entire simulation—from the initial network admittance matrix, through the symbolic structure mapping, to the final linear solver backend. 
This continuous format consistency completely circumvents expensive matrix transposition overheads within the \revision{iteration loop}.

\subsection{The Symbolic Phase}
Algorithm~\ref{alg:symbolic} formalizes this one-time symbolic pass, tracing the exact coordinate and offset allocation required to construct the locked sparsity pattern of $\bm{J}_{\mathrm{red}}$. 

The execution progresses sequentially through the columns of the static network admittance matrix $\bm{Y}_{\mathrm{bus}}$. For each active column $k$ (Line~1), a localized binary search is performed on its sorted row indices (Line~2). Leveraging the block-ordered node layout $[\mathrm{PQ}\mid\mathrm{PV}\mid\mathrm{slack}]$, this bisection splits the column into active and non-active sub-segments, recording the exact boundaries ($\mathit{pq\_end}[k]$, $\mathit{active\_end}[k]$) and caching the absolute memory index of the diagonal entry ($\mathit{diag}[k]$). 

Following the bisection, the active row indices bounded within $[0, n_a)$ are appended directly to the global value layout of $\bm{J}_{\mathrm{red}}$ to instantiate the phase-state ($\theta_k$) column (Line~3). Because the pre-permuted $P\bm{Y}_{\mathrm{bus}}P^{\top}$ topology guarantees that row indices are strictly monotonic, this appending operation translates to a zero-branch memory copy of the active slice, during which the structural boundary offsets for the $J_{11}$ and $J_{21}$ sub-blocks are locked. If the current node satisfies the PQ criteria ($k < N_{\mathrm{pq}}$), this identical row structure is mirrored sequentially to establish the voltage-magnitude ($|V|_k$) column, capturing the corresponding $J_{12}$ and $J_{22}$ offsets (Line~4). Finally, the column structural configuration is closed into the sparse column pointers (Line~5). The resulting static metadata functions as a permanent \textit{symbolic cache}, leaving the subsequent numeric phase entirely free from runtime block-stacking or index-shifting computations.

\begin{algorithm}[t]
\caption{Symbolic phase --- build CSC pattern of $\Jred$}
\label{alg:symbolic}
\begin{algorithmic}[1]
\Require $\Ybus$ (CSC); $N_{pq}$, $n_a=N_{pq}+N_{pv}$; buses $[\mathrm{PQ}\mid\mathrm{PV}\mid\mathrm{slack}]$
\Ensure $\mathit{col\_ptrs}$, $\mathit{row\_idx}$ of $\Jred$; $\mathit{pq\_end}[k]$, $\mathit{active\_end}[k]$, $\mathit{diag}[k]$; block start offsets
\For{$k=0,\dots,n_a-1$}
  \State bisect rows of $\Ybus^{:,k}$ at $N_{pq}$ and $n_a$; record $\mathit{pq\_end}[k]$, $\mathit{active\_end}[k]$, $\mathit{diag}[k]$
  \State append rows $[0,n_a)$ of $\Ybus^{:,k}$ to $\mathit{row\_idx}$ as $\theta_k$ column of $\Jred$; record J$_{11}$, J$_{21}$ start offsets
  \If{$k<N_{pq}$} append same rows as $|V|_k$ column; record J$_{12}$, J$_{22}$ start offsets \EndIf
  \State seal column in $\mathit{col\_ptrs}$
\EndFor
\end{algorithmic}
\end{algorithm}

\subsection{The Direct Single-Pass Numeric Phase}
The per-iteration assembly kernel, formalized as \textsc{Fill} (Algorithm~\ref{alg:fill}), executes as a branch-free, streaming numeric pipeline to fill $\Jred$ directly. 
By utilizing the pre-computed metadata from the symbolic pass (Algorithm~\ref{alg:symbolic}), the numeric kernel completely eliminates runtime memory allocation and spatial coordinate partitioning. 
Leveraging the permutation corollaries derived in Section~\ref{sec:symbolic}, the runtime completely isolates the PQ (Lines~1--5) and PV (Lines~6--10) loops, thereby eliminating all conditional branching inside the hot execution path. 
Because both the input layout of $\bm{Y}_{\mathrm{bus}}$ and the destination offsets of $\bm{J}_{\mathrm{red}}$ are mathematically locked, this element-wise fill operates in a single, direct pass, simultaneously streaming computed sensitivities into multiple distinct matrix memory starting points.

By ensuring contiguous memory access, this algorithm allows standard compilers to automatically inject SIMD vectorization and achieve maximum cache utilization. 
Furthermore, the total elimination of branch divergence and index arithmetic satisfies the core architectural prerequisites for GPU execution, making this numeric kernel natively suited for massively parallel acceleration.

\begin{algorithm}[t]
\caption{\textsc{Fill} --- in-place numeric assembly of $\Jred$}
\label{alg:fill}
\begin{algorithmic}[1]
\Require $\Ybus$, $\Vbus$, $\Vnorm$, $\Ibus=\Ybus\Vbus$; cache from Alg.~\ref{alg:symbolic}
\Ensure value array of $\Jred$ overwritten in place
\For{$k=0,\dots,N_{pq}-1$} \Comment{PQ column: write $\theta_k$ and $|V|_k$}
  \For{$(i,k)\in\mathrm{nz}(\Ybus),\;i<n_a$}
    \State write $\mathcal{U}_\theta(i,k)$, $\mathcal{U}_v(i,k)$~\eqref{eq:coltheta} to $\mathrm{J}_{11}$, $\mathrm{J}_{12}$; also $\mathrm{J}_{21}$, $\mathrm{J}_{22}$ if $i<N_{pq}$
  \EndFor
  \State add $(-Q_k,\,{+}P_k,\,P_k/|V_k|,\,Q_k/|V_k|)$ to diagonal entries~\eqref{eq:kernel}
\EndFor
\For{$k=N_{pq},\dots,n_a-1$} \Comment{PV column: write $\theta_k$ only}
  \For{$(i,k)\in\mathrm{nz}(\Ybus),\;i<n_a$}
    \State write $\mathcal{U}_\theta(i,k)$~\eqref{eq:coltheta} to $\mathrm{J}_{11}$; also $\mathrm{J}_{21}$ if $i<N_{pq}$
  \EndFor
  \State add $-Q_k$ to $\mathrm{J}_{11}$ diagonal entry~\eqref{eq:kernel}
\EndFor
\end{algorithmic}
\end{algorithm}

 \begin{figure*}[t]
  \centering
  \includegraphics[width=2\columnwidth,  trim={0 0.5cm 0 0},
        clip]{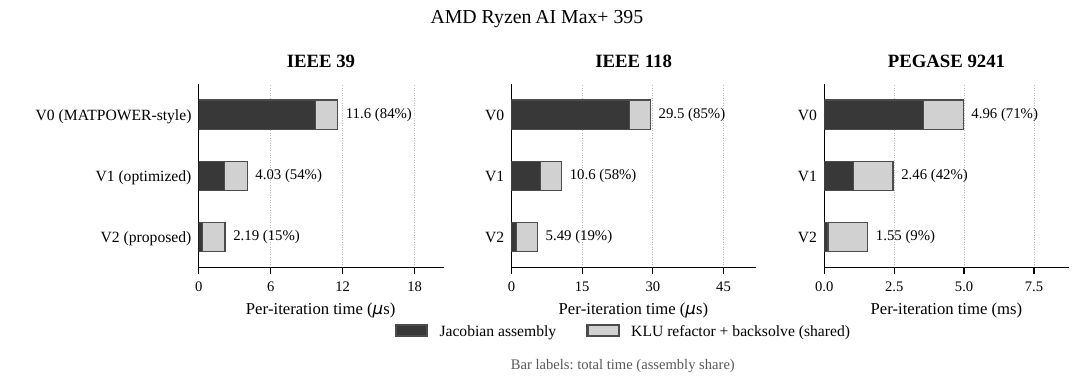}
  \caption{Iteration cost breakdown on the test systems: assembly and factorisation of three methods in \texttt{RustPower}. The proposed method (V2) significantly reduces the assembly share, making factorisation the dominant cost. \revision{Measured on the AMD Ryzen AI Max+~395; factorisation includes shared KLU numeric refactorization and back-substitution, excluding one-time initialization.}}
  \label{fig:breakdown}
\end{figure*}

\begin{algorithm}[t]
\caption{Newton--Raphson with symbolic-cached assembly}
\label{alg:numeric}
\begin{algorithmic}[1]
\Require $\Ybus$, $\Sbus^{\mathrm{sp}}$, $\Vbus$, $N_{pv},N_{pq}$, tolerance $\tau$
\State $\mathit{cache}\gets$ Algorithm~\ref{alg:symbolic} \Comment{symbolic phase, once}
\State allocate the value array of $\Jred$ \Comment{the only Jacobian allocation}
\For{$\nu=0,1,2,\dots$}
  \State $\Ibus\gets\Ybus\Vbus$;\quad $\bm{F}\gets$ power mismatch~\eqref{eq:sbus}
  \If{$\|\bm{F}\|<\tau$} \State \Return $\Vbus$ \EndIf
  \State run Algorithm~\ref{alg:fill} \Comment{overwrite $\Jred$ values in place}
  \State solve $\Jred\,\Delta\bm{x}=-\bm{F}$ \Comment{structure from the cache}
  \State $\Vbus\gets$ voltage update from $\Delta\bm{x}$
\EndFor
\end{algorithmic}
\end{algorithm}
The global integration of these modules is formalized in Algorithm~\ref{alg:numeric}. 
Rather than proposing a modified iterative framework, Algorithm~\ref{alg:numeric} serves as a straightforward replacement for the standard sequence. 
By exposing a unified trait within \texttt{RustPower}, the offline symbolic pass (Algorithm~\ref{alg:symbolic}) and the in-place numeric kernel (Algorithm~\ref{alg:fill}) simply substitute the conventional assembly block, allowing the simulator to interface with interchangeable direct solvers like KLU.

\section{Case Studies}
\label{sec:results}

\subsection{Experimental Setup}
\label{sec:setup}
Performance was measured on an \revision{AMD Ryzen AI Max+~395} CPU with 32~GB memory. Three networks were tested \revision{ for the internal ablation}: the IEEE 39-bus, the IEEE 118-bus~\cite{ieee118bus}, and the PEGASE 9241-bus~\cite{fliscounakis2013contingency} systems. On IEEE 39, transformers are modeled as lines to match pandapower for validation. All can be accessed from the \texttt{RustPower} repository~\cite{rustpower}. \revision{The external comparison additionally includes PEGASE 1354 and uses pandapower network cases.}

We evaluate three paths compiled within the same binary, \revision{using the same KLU backend:}
\begin{itemize}
  \item V0 --- A direct port of MATPOWER. It computes complex matrices~\eqref{eq:dStheta}--\eqref{eq:dSvm} using diagonal products and rebuilds the Jacobian via block stacking on every iteration.
  \item V1 --- An optimized formulation. It uses symbolic permutation from \texttt{RustPower}, a tuned derivative function and cached buffers, but keeps the three-step assembly. It is similar to~\cite{schaefer2018jacobian} but with the pre-permuted partition of PQ,PV nodes.
  \item V2 --- The proposed method of Sections~\ref{sec:symbolic}--\ref{sec:impl}\revision{ in \texttt{RustPower}~0.5.2}.
\end{itemize}
Results are averaged over \revision{seven batches of 300 solves}. We compare these paths against pandapower\revision{~3.5.4 with Numba}~\cite{thurner2018pandapower} and LightSim2Grid\revision{~1.0.0}~\cite{lightsim2grid}. LightSim2Grid uses a C++ backend on a pre-built model. \revision{Given the rapid evolution of power-flow software, comparisons are performed using fixed software releases.}

\revision{The iteration loop timings in Table~\ref{tab:bench} are averages of 20 repeated complete power-flow solves after warm-up, with initialized network models and cached solver data reused and the initial voltage reset for each solve.}

\subsection{Accuracy}
 \revision{All three computational paths (V0, V1, and V2) employ the same Newton-Raphson formulation, with flat-start initialization and the $\ell_\infty$-norm power-mismatch convergence criterion adopted from pandapower. As expected, all solutions require exactly the same number of Newton iterations across every tested case. The computed solutions exhibit near-machine-precision agreement with pandapower, with maximum voltage-magnitude and phase-angle deviations of $3.38\times10^{-13}$ p.u. and $3.92\times10^{-11}$ degrees, respectively, across all tested systems.}. 

\subsection{Internal Ablation: Isolating the Symbolic Gain}
Table~\ref{tab:ablation} reports solve times. Transitioning from V0 to V1 yields a \revision{$2.0\times$ to $2.3\times$} speedup via memory tuning. Transitioning from V1 to V2 provides an additional \revision{$1.6\times$ to $2.0\times$} speedup. Because these totals include the same solve, they understate the assembly speedup.

Table~\ref{tab:asmonly} isolates the assembly step. Here, the symbolic assembly is \revision{$6.0\times$ to $7.8\times$} faster than the tuned path, and \revision{$24.4\times$ to $30.1\times$} faster than the V0 port. The speedup remains independent of the network size.

Profiling a single iteration reveals the runtime allocation (Fig.~\ref{fig:breakdown}). In the standard port, assembly consumes \revision{$84\%$, $85\%$, and $71\%$} of the iteration time. Tuning reduces this share to \revision{$54\%$, $58\%$, and $42\%$}. The proposed method eliminates the structural bottleneck, dropping the assembly share to \revision{$15\%$, $19\%$, and $9\%$}. Consequently, numeric factorization becomes the primary cost.

This speedup stems from structural simplification. Traditional assembly traverses nonzeros repeatedly to form intermediate matrices, sub-blocks, and the stacked Jacobian. The symbolic pass eliminates intermediates. By pre-mapping where each nonzero lands, the numeric phase reads the admittance matrix once and writes the Jacobian once.
\begin{table}[t]
\caption{Ablation of Jacobian assembly: solve time (ms) with same solver.}
\label{tab:ablation}
\centering
\setlength{\tabcolsep}{3pt}
\begin{tabular*}{\columnwidth}{@{\extracolsep{\fill}}lrrrrr@{}}
\toprule
System & V0 & V1 & V2 & V0/V2 & V1/V2 \\
\midrule
IEEE 39 & \revision{$0.0500$} & \revision{$0.0228$} & \revision{$\bm{0.0114}$} & \revision{$4.4\times$} & \revision{$2.0\times$} \\
IEEE 118 & \revision{$0.1253$} & \revision{$0.0540$} & \revision{$\bm{0.0277}$} & \revision{$4.5\times$} & \revision{$2.0\times$} \\
PEGASE 9241 & \revision{$35.372$} & \revision{$17.850$} & \revision{$\bm{11.125}$} & \revision{$3.2\times$} & \revision{$1.6\times$} \\
\bottomrule
\end{tabular*}
\end{table}

\begin{table}[t]
\caption{Isolated assembly step: assembly time ($\mu$s) excluding solve.}
\label{tab:asmonly}
\centering
\setlength{\tabcolsep}{3pt}
\begin{tabular*}{\columnwidth}{@{\extracolsep{\fill}}lrrrrr@{}}
\toprule
System & V0 & V1 & V2 & V0/V2 & V1/V2 \\
\midrule
IEEE 39 & \revision{$9.707$} & \revision{$2.169$} & \revision{$\bm{0.323}$} & \revision{$30.1\times$} & \revision{$6.7\times$} \\
IEEE 118 & \revision{$24.998$} & \revision{$6.153$} & \revision{$\bm{1.023}$} & \revision{$24.4\times$} & \revision{$6.0\times$} \\
PEGASE 9241 & \revision{$3547.383$} & \revision{$1040.267$} & \revision{$\bm{134.072}$} & \revision{$26.5\times$} & \revision{$7.8\times$} \\
\bottomrule
\end{tabular*}
\end{table}

\subsection{Comparison with Established Tools}

Table~\ref{tab:bench} compares \revision{the measured times of pandapower~3.5.4, LightSim2Grid~1.0.0, and \texttt{RustPower}~0.5.2. RustPower reports both grid execution with post-processing and solver-core time of iteration loops.}

\revision{Including post-processing, \texttt{RustPower}'s speedup over pandapower~3.5.4 with Numba falls from $154.3\times$ on IEEE~39 to $11.7\times$ on PEGASE~9241 (Table~\ref{tab:bench}); Amdahl's law limits the gain from faster assembly as factorization takes a larger share on larger systems, however, it still provides significant speed-up.  }

\revision{In Table~\ref{tab:bench}, ``RP grid'' includes dynamic Bevy ECS system execution and post-processing of bus voltages and powers and line flows, whereas LightSim2Grid operates directly on arrays; the additional time beyond ``RP core'' decreases from 231~ns per bus on IEEE~39 to 19~ns per bus on PEGASE~9241 and only takes $1.7\%$ of total execution time on PEGASE~9241.}

\begin{table}[t]
\caption{\revision{External Tools Iteration Loop Comparison (ms)}}
\label{tab:bench}
\centering
\small
\setlength{\tabcolsep}{1.75pt}
\begin{tabular*}{\columnwidth}{@{\extracolsep{\fill}}lrrrrrr@{}}
\toprule
System & PP & LS2G & \revision{RP grid} & \revision{RP core} & vs LS2G & vs PP \\
\midrule
IEEE 39 & \revision{$3.241$} & \revision{$0.020$} & \revision{$0.021$} & \revision{$0.012$} & \revision{$0.95\times$} & \revision{$154.3\times$} \\
IEEE 118 & \revision{$4.112$} & \revision{$0.052$} & \revision{$0.039$} & \revision{$0.028$} & \revision{$1.33\times$} & \revision{$105.4\times$} \\
\revision{PEGASE 1354} & \revision{$15.830$} & \revision{$1.041$} & \revision{$0.589$} & \revision{$0.554$} & \revision{$1.77\times$} & \revision{$26.9\times$} \\
PEGASE 9241 & \revision{$126.408$} & \revision{$13.766$} & \revision{$10.804$} & \revision{$10.625$} & \revision{$1.27\times$} & \revision{$11.7\times$} \\
\bottomrule
\end{tabular*}

\smallskip
{\footnotesize PP = pandapower\revision{~3.5.4 (Numba)}; LS2G = LightSim2Grid\revision{~1.0.0}; \revision{RP = RustPower 0.5.2; grid includes post-processing. Speedups are reference time divided by RP grid time.}}
\end{table}
\section{Conclusion}
\label{sec:conclusion}

This paper delivers three contributions for steady-state power-flow analysis. First, we mathematically prove the permutation equivariance of complex sensitivity, verifying that symbolic permutation pre-applied to the admittance matrix eliminates runtime projections. Second, we develop a generic algorithm that couples symbolic permutation with branchless filling to assemble the global Jacobian in a single pass. Third, we realize this framework in the \texttt{RustPower} solver. Benchmarks demonstrate significant acceleration of power-flow computations, \revision{delivering a $3.2\times$ to $4.5\times$ solve speedup over V0 and $1.6\times$ to $2.0\times$ over V1, while reducing assembly overhead by $24.4$--$30.1\times$ and $6.0$--$7.8\times$, respectively} without altering the numeric results or the convergence.

The proposed methodology, the generic algorithm, and the \texttt{RustPower} toolchain establish a new foundation for modern power system analysis. The new method has potential for faster optimal power flow, pre-conditioning, and continuous Newton power-flow solvers. It may also be applicable to transient stability simulation, where streaming logic provides optimal performance on multi-GPU architectures~\cite{jalilimarandi2012gpu} and even FPGAs~\cite{shiqi}. Ultimately, while artificial intelligence increases the demand for advanced computation, the proposed work helps bridge the widening gap between evolving computing paradigms and traditional analysis tools for modern power systems.
\bibliographystyle{IEEEtran}
\bibliography{./bib/IEEEabrv,./bib/rustpower}

@STRING{IEEE_J_PWRS       = "{IEEE} Trans. Power Syst."}

@book{bergen2000power,
  author    = {Bergen, A. R. and Vittal, V.},
  title     = {Power Systems Analysis},
  edition   = {2nd},
  publisher = {Prentice Hall},
  year      = {2000}
}

@article{tinney1967power,
  author  = {Tinney, W. F. and Hart, C. E.},
  title   = {Power Flow Solution by {Newton's} Method},
  journal = {IEEE Trans. Power App. Syst.},
  year    = {1967},
  volume  = {PAS-86},
  number  = {11},
  pages   = {1449--1460},
  doi     = {10.1109/TPAS.1967.291823}
}

@article{tinney1967direct,
  author  = {Tinney, W. F. and Walker, J. W.},
  title   = {Direct Solutions of Sparse Network Equations by Optimally Ordered Triangular Factorization},
  journal = {Proc. IEEE},
  year    = {1967},
  volume  = {55},
  number  = {11},
  pages   = {1801--1809},
  doi     = {10.1109/PROC.1967.6011}
}

@article{davis2010klu,
  author  = {Davis, T. A. and Natarajan, E. P.},
  title   = {Algorithm 907: {KLU}, A Direct Sparse Solver for Circuit Simulation Problems},
  journal = {{ACM} Trans. Math. Softw.},
  volume  = {37},
  number  = {3},
  pages   = {36:1--36:17},
  year    = {2010},
  doi     = {10.1145/1824801.1824814}
}

@misc{rustpower,
  author       = {Cheng, Tianshi},
  title        = {{RustPower}: A Data-Oriented Open-Source Power-Flow Solver in {Rust}},
  howpublished = {\url{https://github.com/chengts95/rustpower}},
  year         = {2024}
}

@article{zimmerman2011matpower,
  author  = {Zimmerman, R. D. and Murillo-Sanchez, C. E. and Thomas, R. J.},
  title   = {{MATPOWER}: Steady-State Operations, Planning, and Analysis Tools for Power Systems Research and Education},
  journal = IEEE_J_PWRS,
  year    = {2011},
  volume  = {26},
  number  = {1},
  pages   = {12--19},
  doi     = {10.1109/TPWRS.2010.2051168}
}

@article{thurner2018pandapower,
  author  = {Thurner, L. and Scheidler, A. and Schäfer, F. and Menke, J. and Dollichon, J. and Meier, F. and Meinecke, S. and Braun, M.},
  title   = {pandapower --- {An} Open-Source {Python} Tool for Convenient Modeling, Analysis, and Optimization of Electric Power Systems},
  journal = IEEE_J_PWRS,
  year    = {2018},
  volume  = {33},
  number  = {6},
  pages   = {6510--6521},
  doi     = {10.1109/TPWRS.2018.2829021}
}

@inproceedings{schaefer2018jacobian,
  author    = {Sch{\"a}fer, Florian and Braun, Martin},
  title     = {An Efficient Open-Source Implementation to Compute the Jacobian Matrix for the {Newton-Raphson} Power Flow Algorithm},
  booktitle = {Proc. IEEE PES Innov. Smart Grid Technol. Conf. Europe (ISGT-Europe)},
  year      = {2018},
  pages     = {1--6},
  note      = {arXiv:1804.06742},
  doi       = {10.1109/ISGTEurope.2018.8571471}
}

@misc{lightsim2grid,
  author       = {Donnot, Benjamin},
  title        = {{LightSim2Grid}: A {C++} Backend Targeting the {Grid2Op} Platform for Power-System Simulation},
  howpublished = {\url{https://github.com/BDonnot/lightsim2grid}},
  year         = {2020}
}

@misc{pypower,
  author       = {Lincoln, Richard},
  title        = {{PYPOWER}: A {Python} Port of {MATPOWER}},
  howpublished = {\url{https://pypi.python.org/pypi/PYPOWER}},
  year         = {2015}
}

@INPROCEEDINGS{ryu_gpuopf,
  author={Ryu, Minseok and Byeon, Geunyeong and Kim, Kibaek},
  booktitle={Proc. IEEE Int. Parallel Distrib. Process. Symp. (IPDPS)},
  title={A {GPU}-Accelerated Distributed Algorithm for Optimal Power Flow in Distribution Systems},
  year={2025},
  volume={},
  number={},
  pages={702-711},
  doi={10.1109/IPDPS64566.2025.00068}}

@ARTICLE{zhou_gpuPPF,
  author={Zhou, Gan and Bo, Rui and Chien, Lungsheng and Zhang, Xu and Yang, Shengchun and Su, Dawei},
  journal={IEEE Trans. Power Syst.},
  title={{GPU}-Accelerated Algorithm for Online Probabilistic Power Flow},
  year={2018},
  volume={33},
  number={1},
  pages={1132-1135},
  doi={10.1109/TPWRS.2017.2756339}}

@ARTICLE{zhou2017gpu,
  author={Zhou, Gan and Feng, Yanjun and Bo, Rui and Chien, Lungsheng and Zhang, Xu and Lang, Yansheng and Jia, Yupei and Chen, Zhengping},
  journal={IEEE Trans. Smart Grid},
  title={{GPU}-Accelerated Batch-{ACPF} Solution for {N-1} Static Security Analysis},
  year={2017},
  volume={8},
  number={3},
  pages={1406-1416},
  doi={10.1109/TSG.2016.2600587}}

@article{fliscounakis2013contingency,
  author  = {Fliscounakis, S. and Panciatici, P. and Capitanescu, F. and Wehenkel, L.},
  title   = {Contingency Ranking With Respect to Overloads in Very Large Power Systems Taking Into Account Uncertainty, Preventive, and Corrective Actions},
  journal = IEEE_J_PWRS,
  year    = {2013},
  volume  = {28},
  number  = {4},
  pages   = {4909--4917},
  doi     = {10.1109/TPWRS.2013.2251015}
}

@techreport{ieee118bus,
  author      = {{University of Washington}},
  title       = {{IEEE} 118-Bus Test Case},
  institution = {Power Systems Test Case Archive, University of Washington},
  year        = {1993},
  url         = {https://labs.ece.uw.edu/pstca/pf118/pg_tca118bus.htm}
}

@misc{bevyengine,
  author       = {{Bevy Contributors}},
  title        = {Bevy: A Data-Driven Game Engine Built in {Rust}},
  howpublished = {\url{https://bevyengine.org}},
  year         = {2020}
}

@inproceedings{nystrom2014data,
  author    = {Nystrom, Robert},
  title     = {Data Locality (Game Programming Patterns)},
  booktitle = {Game Programming Patterns},
  year      = {2014},
  publisher = {Genever Benning},
  url       = {https://gameprogrammingpatterns.com/data-locality.html}
}

@article{cheng2023ecsgrid,
  author  = {Cheng, Tianshi and Duan, Tong and Dinavahi, Venkata},
  title   = {{ECS-Grid}: Data-Oriented Real-Time Simulation Platform for Cyber-Physical Power Systems},
  journal = {IEEE Trans. Ind. Informat.},
  year    = {2023},
  volume  = {19},
  number  = {11},
  pages   = {11128--11138},
  doi     = {10.1109/TII.2023.3244329}
}

@book{matsakis2019rust,
  author    = {Klabnik, Steve and Nichols, Carol},
  title     = {The {Rust} Programming Language},
  publisher = {No Starch Press},
  year      = {2019},
  url       = {https://doc.rust-lang.org/book/}
}

@techreport{drepper2007what,
  author      = {Drepper, Ulrich},
  title       = {What Every Programmer Should Know About Memory},
  institution = {Red Hat, Inc.},
  year        = {2007},
  url         = {https://www.akkadia.org/drepper/cpumemory.pdf}
}

@article{milano2009open,
  author  = {Milano, F.},
  title   = {An Open Source Power System Analysis Toolbox},
  journal = IEEE_J_PWRS,
  year    = {2005},
  volume  = {20},
  number  = {3},
  pages   = {1199--1206},
  doi     = {10.1109/TPWRS.2005.851911}
}

@ARTICLE{shiqi,
  author={Cao, Shiqi and Lin, Ning and Dinavahi, Venkata},
  journal={IEEE Trans. Power Syst.},
  title={Faster-Than-Real-Time Hardware Emulation of Extensive Contingencies for Dynamic Security Analysis of Large-Scale Integrated {AC/DC} Grid},
  year={2023},
  volume={38},
  number={1},
  pages={861-871},
  doi={10.1109/TPWRS.2022.3161561}}

@article{jalilimarandi2012gpu,
  author  = {Jalili-Marandi, V. and Zhou, Z. and Dinavahi, V.},
  title   = {Large-Scale Transient Stability Simulation of Electrical Power Systems on Multiple {GPU}s},
  journal = {IEEE Trans. Parallel Distrib. Syst.},
  year    = {2012},
  volume  = {23},
  number  = {7},
  pages   = {1255--1265},
  doi     = {10.1109/TPDS.2011.220}
}

@article{wang2021fast,
  author    = {Zhenqi Wang and Sebastian Wende-von Berg and Martin Braun},
  title     = {Fast parallel {Newton--Raphson} power flow solver for large number of system calculations with {CPU} and {GPU}},
  journal   = {Sustain. Energy Grids Networks},
  volume    = {27},
  pages     = {100483},
  year      = {2021},
  issn      = {2352-4677},
}

@article{iwamoto1981load,
  author  = {Iwamoto, M. and Tamura, Y.},
  title   = {A Load Flow Calculation Method for Ill-Conditioned Power Systems},
  journal = {IEEE Trans. Power App. Syst.},
  year    = {1981},
  volume  = {PAS-100},
  number  = {4},
  pages   = {1736--1743},
  doi     = {10.1109/TPAS.1981.316511}
}

@ARTICLE{yu2025cfn,
  author={Yu, Ruizhi and Gu, Wei and Xu, Yijun and Lu, Shuai and Zhang, Suhan},
  journal={CSEE J. Power Energy Syst.}, 
  title={Semi-Implicit Continuous Newton Method for Power Flow Analysis}, 
  year={2025},
  volume={11},
  number={4},
  pages={1957-1961},
  doi={10.17775/CSEEJPES.2024.05770}}

@article{idema2012scalable,
  author  = {Idema, R. and Lahaye, D. J. P. and Vuik, C. and van der Sluis, L.},
  title   = {Scalable {Newton-Krylov} Solver for Very Large Power Flow Problems},
  journal = {IEEE Trans. Power Syst.},
  year    = {2012},
  volume  = {27},
  number  = {4},
  pages   = {2248--2257},
  doi     = {10.1109/TPWRS.2012.2190078}
}

@article{yetkin2023recycling,
  author  = {E. Fatih Yetkin and Ceylan, Oguzhan},
  title   = {Recycling {Newton--Krylov} Algorithm for Efficient Solution of Large Scale Power Systems},
  journal = {Int. J. Electr. Power Energy Syst.},
  volume  = {144},
  pages   = {108559},
  year    = {2023},
  issn    = {0142-0615},
  doi     = {10.1016/j.ijepes.2022.108559}
}

@article{pires2019robust,
  author  = {Pires, Robson and Mili, Lamine and Chagas, Guilherme},
  title   = {Robust Complex-Valued {Levenberg-Marquardt} Algorithm as Applied to Power Flow Analysis},
  journal = {Int. J. Electr. Power Energy Syst.},
  volume  = {113},
  pages   = {383--392},
  year    = {2019},
  issn    = {0142-0615},
  doi     = {10.1016/j.ijepes.2019.05.032}
}

@inproceedings{lagace2008improving,
  author    = {Lagace, P. J. and Vuong, M. H. and Kamwa, I.},
  booktitle = {Proc. IEEE Power Energy Soc. Gen. Meeting},
  title     = {Improving Power Flow Convergence by {Newton-Raphson} with a {Levenberg-Marquardt} Method},
  year      = {2008},
  pages     = {1--6},
  doi       = {10.1109/PES.2008.4596138}
}
    \vspace{-0.5in}
    \begin{IEEEbiography}[{\includegraphics[width=1in,height=1.25in,clip,keepaspectratio]{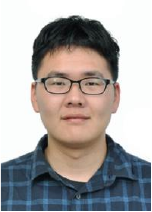}}]{Tianshi Cheng}(S'19-M'24)
    received the B.Eng. degree in electrical engineering and automation from Southeast University, China, in 2017 , and the Ph.D. degree in electrical and computer engineering from the University of Alberta, Canada, in 2024. From 2017 to 2018, he was a substation automation engineer at NARI Group Corporation (State Grid Electric Power Research Institute), China. He is currently working as a Simulation Specialist at RTDS Technologies Inc., Winnipeg, MB, Canada. His research interests include electromagnetic transient simulation, transient stability analysis, steady-state analysis, real-time hardware-in-the-loop simulation, parallel processing, microgrids, and power electronics.
    \end{IEEEbiography}
     \vspace{-0.5in}
    \begin{IEEEbiography}[{\includegraphics[width=1in,height=1.25in, clip,keepaspectratio]{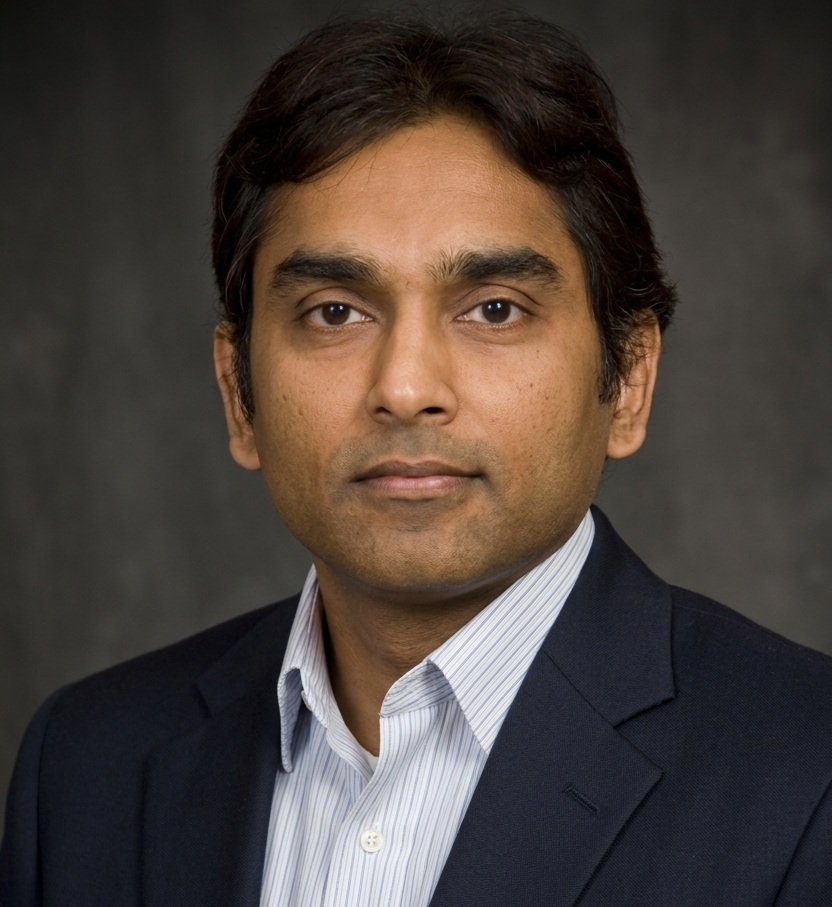}}]{Venkata Dinavahi}(Fellow, IEEE) received the B.Eng. degree in electrical engineering from the Visvesvaraya National Institute of Technology, Nagpur, India, the M.Tech. degree in electrical engineering from the Indian Institute of Technology, Kanpur, India, and the Ph.D. degree in electrical and computer engineering from the University of Toronto, Ontario, Canada. He is currently a Professor with the Department of Electrical and Computer Engineering, University of Alberta, Edmonton, Canada. His research interests include real-time simulation of multi-domain sustainable energy systems, electromagnetic transients, device-level modeling, artificial intelligence, large-scale systems, and parallel and distributed computing. Dr. Dinavahi is a Fellow of the Engineering Institute of Canada (EIC), the Asia-Pacific Artificial Intelligence Association (AAIA), the International Artificial Intelligence Industry Alliance (AIIA). He is an honorary member of the National Academy of Artificial Intelligence (NAAI), and a Professional Engineer in the Province of Alberta, Canada.
    \end{IEEEbiography}
\end{document}